\documentclass{amsart}

\usepackage{amsthm,amsmath,latexsym,amssymb,amsfonts,amscd,hyperref}
\usepackage[all]{xy}
\newtheorem{theorem}{Theorem}[section]
\newtheorem{lemma}[theorem]{Lemma}
\newtheorem{proposition}[theorem]{Proposition}

\theoremstyle{definition}

\newtheorem{example}[theorem]{Example}

\theoremstyle{remark}
\newtheorem{remark}[theorem]{Remark}

\numberwithin{equation}{section}

\begin{document}
\begin{flushright}
{LMU-ASC 30/17}
\end{flushright}
\vspace{2cm}
\title[Hochschild cohomology of the Weyl algebra]{Hochschild cohomology of the Weyl algebra \\and Vasiliev's equations}

\author{Alexey A. Sharapov}
\address{Physics Faculty, Tomsk State University, Lenin ave. 36, Tomsk 634050, Russia}
\email{sharapov@phys.tsu.ru}
\thanks{The first author was supported in part by RFBR
Grant No. 16-02-00284 A}

\author{Evgeny D. Skvortsov}
\address{Arnold Sommerfeld Center for Theoretical Physics,
Ludwig-Maximilians University Munich,
Theresienstr. 37, D-80333 Munich, Germany}
\address{Lebedev Institute of Physics,
Leninsky ave. 53, 119991 Moscow, Russia}
\email{evgeny.skvortsov@lmu.de}
\thanks{The second author was supported in part by the Russian Science Foundation grant 14-42-00047
in association with Lebedev Physical Institute and by the DFG Transregional Collaborative
Research Centre TRR 33 and the DFG cluster of excellence ``Origin and Structure of the Universe".}





\begin{abstract}
We propose a simple injective resolution for the Hochschild complex of the Weyl algebra. By making use of this resolution, we derive explicit expressions for nontrivial cocycles of the Weyl algebra with coefficients in twisted bimodules as well as for the smash products of the Weyl algebra and a  finite group of linear symplectic transformations.  A relationship with the higher-spin field theory is briefly discussed.
\end{abstract}

\maketitle

\section{Introduction}

The polynomial Weyl algebra $A_n$ is defined to be an associative,
unital algebra over $\mathbb{C}$ on $2n$ generators $y^j$ subject to
the relations
$$
y^j y^k-y^k y^j = 2i\pi^{jk}.
$$
Here $\pi=(\pi^{jk})$ is a nondegenerate, anti-symmetric matrix over $\mathbb{C}$. Bringing the matrix $\pi$ into the canonical form, one can see that
$A_n\simeq A_1^{\otimes n}$.

The Hochschild (co)homology of the Weyl algebra is usually computed employing a Koszul-type resolution, see e.g. \cite{AFLS}, \cite{Kassel},  \cite{Pinczon}.
More precisely, the Koszul complex of the Weyl algebra is defined as a finite subcomplex of the normalized bar-resolution, so that the restriction map induces an isomorphism in (co)homology \cite{Pinczon}. This makes it relatively easy to calculate the dimensions of the cohomology spaces for various coefficients.
Finding explicit expressions for nontrivial cocycles appears to be a much more difficult problem. For example, it had long been known that
$$
HH^\bullet(A_n,A^\ast_n)\simeq HH^{2n}(A_n,A_n^\ast)\simeq \mathbb{C}\,.
$$
(The only nonzero group is dual to the homology group $HH_{2n}(A_{n},A_{n})$ generated by the cycle $1\otimes y^1\wedge y^2\wedge\cdots\wedge y^{2n}$.)
An explicit formula for a nontrivial $2n$-cocycle $\tau_{2n}$, however, remained unknown, until it was obtained in 2005 paper by Feigin, Felder and Shoikhet \cite{FFS} as a consequence of  the Tsygan formality for chains \cite{Shoikhet:2000gw}, \cite{Dolgushev}. The cocycle $\tau_{2n}$ was then  used to define a canonical trace in deformation quantization of symplectic manifolds. An earlier discussion of the analytical structure of $\tau_{2n}$ can be found in \cite[Sec. 3.3]{FeiginTsygan}.

In this paper, we propose a simple injective resolution for the Hochschild complex of the Weyl algebra.
This resolution allows us to derive explicit formulae for nontrivial Hochschild cocycles of the Weyl algebra, much as the Koszul resolution provides us with
information about the dimensions of cohomology spaces. Furthermore, the construction naturally extends to the smash products $A_n\rtimes G$, with $G$ being a finite subgroup of $Sp(n,\mathbb{C})$.  Our interest to the problem is mostly motivated by developments in the higher-spin field theory, see \cite{Vasiliev:1999ba}, \cite{Didenko:2014dwa}, \cite{Sharapov:2017yde} for a review. Let us dwell on this point in more detail.

The nonlinear equations of motion governing the dynamics of massless higher-spin fields were proposed by Vasiliev in the late eighties \cite{Vasiliev:1988sa}.  Unlike what we are accustomed  to in the Yang--Mills theory or gravity, the gauge algebra underlying the massless higher-spin fields is an associative algebra rather than a Lie algebra. In the case of four-dimensional space-time, for example, this algebra, called the higher-spin algebra, is given by the smash product $\mathfrak{A}=A_2 \rtimes G$, where $G$ is a certain finite subgroup of $Sp(2,\mathbb{C})$ (see Example \ref{Ex} below). It turns out that the nontrivial interaction vertices making the field equations  nonlinear are completely determined by nontrivial two-cocycles of the algebra $\mathfrak{A}$ with coefficients in itself. These two-cocycles, in turn, are expressed in terms of the Hochschild cocycle $\tau_2$ of the Weyl algebra $A_1$. An explicit formula for the cocycle $\tau_2$ was thus found in \cite{Vasiliev:1988sa} long before the paper \cite{FFS}.

It should be noted that the original Vasiliev's equations were not written in a closed form, rather they involved an infinite sequence of vertices to be found successively by means of homological perturbation theory.  Later on Vasiliev has been  able to give his equation quite a simple and closed form \cite{Vasiliev:1990en} by embedding the higher-spin algebra $\mathfrak{A}$ into a bigger algebra, called the Vasiliev double \cite{Sharapov:2017yde}.  The initial field equations arise then after excluding the auxiliary fields together with the additional algebra generators. As by product, one obtains an explicit expression for the Hochschild cocycle $\tau_2$ determining the interaction vertices. A closer look at this doubling/exclusion procedure
shows that it implicitly uses the general algebraic concept of resolution. The aim of this paper is to present this resolution in an explicit form, free from all irrelevant field-theoretical details.

\section{Feigin--Felder--Shoikhet Cocycle}
\label{sec:FFS}

In the Introduction, we have defined the Weyl algebra in terms of generators and relations. Alternatively, one can think of the algebra $A_n$ as the space of complex polynomials $\mathbb{C}[y^1,\ldots, y^{2n}]$ endowed with the Moyal $\ast$-product:
$$
a\ast b = a\exp i\left(\frac{\overleftarrow{\partial}}{\partial y^j}\pi^{jk}\frac{\overrightarrow{\partial}}{\partial y^k}\right)b\,.
$$
The Weyl algebra enjoys the involutive automorphism:
\begin{equation}\label{inv}
    a\mapsto \tilde{a}\,,\qquad \tilde{a}(y)=a(-y)\,.
\end{equation}
This allows one to equip $A_n$ with the $\mathbb{Z}_2$-grading: $A_n=A_n^0\oplus A_n^1$, where the even and odd subspaces are generated, respectively,  by the even and odd polynomials in $y$'s.  As a $\mathbb{Z}_2$-graded algebra, $A_n$ admits a supertrace $\mathrm{str}: A_n\rightarrow \mathbb{C}$. By definition,
$$
    \mathrm{str}(a)=a(0)\,,\qquad \mathrm{str}(a\ast b)=(-1)^{|a||b|}\mathrm{str}(b\ast a)\,,
$$
where $|a|$ stands for the parity of $a\in A_n$.
Canonically associated to the supertrace is the supersymmetric bilinear form
\begin{equation}\label{B}
B(a,b)=\mathrm{str}(a\ast b)\,,\qquad B(a,b)=(-1)^{|a||b|}B(b,a)\,.
\end{equation}
The form $B$ is known to be nondegenerate \cite{Pinczon2005}, see also \cite{Konstein:2012xp}. One can check that
\begin{equation}\label{ab}
    B(a, b)=B(\tilde{b},  a)
\end{equation}
and the subspaces $A_n^0$ and $A_n^1$ are orthogonal to each other. Using the nondegenerate bilinear form $B$, we can identify $A_n$ with a subspace in its dual space $A_n^\ast=\mathbb{C}[[y^1,\ldots,y^{2n}]]$. The latter, by definition, consists of the formal power series in $y$'s with complex coefficients.  The natural $A_n$-bimodule structure on $A_n$ induces that on the dual space $A_n^\ast$. It follows form (\ref{ab}) that the left and right actions of $A_n$ on $A_n^\ast$ are given by
$$
    a\mapsto b\ast a\ast \tilde{c}\,,\qquad \forall a\in A_n^\ast,\quad \forall b,c\in A_n\,.
$$
Notice that both the $\ast$-products  are well defined.

Now we turn to the cohomology of the Weyl algebra. Recall that the Hochschild cohomology $HH^\bullet(A,M)$ of an associative unital $k$-algebra $A$  with coefficients in a bimodule $M$ over $A$ is the cohomology of the Hochschild cochain complex \cite{CartEil}
$$
\xymatrix{C^\bullet(A,M) :   &C^0\ar[r]^{\partial}& C^1\ar[r]^{\partial}& C^2\ar[r]^{\partial} &\cdots}
$$
with
$$
C^p=\mathrm{Hom}_{k}(A^{\otimes p}, M)\,,\qquad A^{\otimes p}=\underbrace{A\otimes \cdots \otimes A}_p\,,
$$
and the differential
$$
(\partial f)(a_1,\ldots, a_{p+1})=a_1 f(a_2,\ldots,a_{p+1})+\sum_{k=1}^{p}(-1)^{k+1} f(a_1,\ldots,a_ka_{k+1},\ldots, a_{p+1})
$$
$$
+(-1)^{p+1}f(a_1,\ldots,a_{p})a_{p+1}\,.
$$

The complex $C^
\bullet(A,M)$ contains a large subcomplex $\bar C^\bullet(A,M)$ of cochains that vanish when at least one of their arguments is equal to $1\in A$. The latter is called the {normalized Hochschild complex}. It is easy to see that the inclusion map $i: \bar C(A,M)\rightarrow C(A,M)$ induces an isomorphism in cohomology, meaning that the Hochschild cohomology of $A$ is isomorphic to that of the quotient  algebra  $\bar A=A/\mathbb{C}1$.

Of particular interest are two special cases of modules: $M=A$ and $M=A^\ast$. The cohomology groups $HH^\bullet(A,A)$ have been extensively studied because of their relation to deformation theory, while the groups $HH^\bullet(A,A^\ast)$ are functorial in the $k$-algebra $A$. For the Weyl algebra, we have
\begin{equation}\label{HH0}
HH^\bullet(A_n,A_n)\simeq HH^0(A_n,A_n)\simeq \mathbb{C}\,,
\end{equation}
\begin{equation}\label{HH}
HH^\bullet(A_n,A^\ast_n)\simeq HH^{2n}(A_n,A^\ast_n)\simeq \mathbb{C}\,.
\end{equation}
The group $HH^0(A_n,A_n)$ is identified with the center of the Weyl algebra and the Feigin--Felder--Shoikhet (FFS) cocycle we are interested in generates the other nontrivial group $HH^{2n}(A_n,A^\ast_n)$.

Let $\Delta_{2n}$ be the standard simplex in $\mathbb{R}^{2n}$,
$$
\Delta_{2n} :\quad 0=u_0\leq u_1\leq \cdots\leq u_{2n}\leq 1\,,
$$
and introduce the notation
$$
    p_0^j=i y^j\,,\qquad p_\mu^j=\pi^{jk}\frac{\partial}{\partial y_\mu^k}\,,\qquad \mu=1,\ldots,2n\,.
$$
Then the FFS cocycle can be written as
\begin{equation}\label{FFS}
    \tau_{2n}(a_1,\ldots, a_{2n})=\hat{\tau}_{2n}(p_0, p_1,\ldots,p_{2n}) a_1(y_1)a_2(y_2)\cdots a_{2n}(y_{2n})|_{y_\mu=0}\,,
\end{equation}
where $a_\mu \in A_n$ and
\begin{equation}\label{symb}
\begin{array}{c}
   \hat \tau_{2n}(p_0, p_1,\ldots,p_{2n})= \displaystyle \det(p_1,\ldots,p_{2n})\\[3mm]
   \displaystyle\times\int_{\Delta_{2n}} \exp{i\left[\sum_{0\leq i<j\leq 2n} (1+2u_i-2u_j)\omega(p_i, p_j)\right]}d^{2n}u
   \end{array}
\end{equation}
is the symbol of the polydifferential operator in (\ref{FFS}). Here
$\omega(p,q)=\omega_{ij}p^i q^j$ and $\omega=(\omega_{ij})$ is the
two-form dual to $\pi=(\pi^{ij})$. The exponential function in the
integral is to be expanded in the Taylor series and integrated term
by term. (As the functions $a_\mu$ are polynomial, only finitely
many terms contribute nontrivially to (\ref{FFS}).)

\begin{theorem}[FFS] The cochain $\tau_{2n}$ is a nontrivial cocycle in $\bar C^{2n}(A_n,A^\ast_n)$.

\end{theorem}
The proof can be found in \cite{FFS}. The idea is to represent the
l.h.s. of the cocycle condition
$$
a_0\ast \tau_{2n}(a_1,\ldots,a_{2n})-\tau_{2n}(a_0\ast a_1,\ldots,a_{2n})+\ldots -\tau_{2n}(a_0,\ldots,a_{2n-1})\ast\tilde{a}_{2n}=0
$$
as an integral over the faces of a $(2n+1)$-dimensional simplex and to apply then Stokes' theorem.
One more geometrical interpretation of the cocycle condition is presented in Appendix A.
Evaluating now the cocycle $\tau_{2n}$ on the basis cycle $c_{2n}=1\otimes y^1\wedge \cdots\wedge y^{2n}$, one can find
\begin{equation}\label{nc}
B(1,\tau_{2n}(y^1\wedge \cdots \wedge
y^{2n}))=\tau_{2n}(y^1\wedge\cdots\wedge y^{2n})(0)=(2n!)^{-1}\,.
\end{equation}
This shows that both $\tau_{2n}$ and $c_{2n}$ are nontrivial.

\begin{remark}
The fact that the chain $c_{2n}$ is a Hochschild cycle is not of crucial importance in the proof above.
Indeed, for any normalized cochain $\gamma\in \bar C^{2n-1}(A_n,A_n^\ast)$ one can find
$$
(\partial\gamma)(y^1\wedge\cdots\wedge
y^{2n})=2\sum_{k=1}^{2n}(-1)^{k+1}y^{k} \gamma(y^1\wedge\cdots\wedge\hat{y}^k\wedge\cdots\wedge
y^{2n})\,.
$$
(Due to the involution (\ref{inv}) twisting the right action of $A_n$ on $A_n^\ast$, the Moyal $\ast$-product of $y$'s and $\gamma$ reduces effectively to the ordinary, i.e., commutative
multiplication with overall factor $2$.) If the FFS cocycle were
trivial, i.e., $\tau_{2n}=\partial \gamma$, then we would have
$\tau_{2n}(y^1\wedge \cdots\wedge y^{2n})(0)=0$, which is not the
case as seen from (\ref{nc}).
\end{remark}

The integral over the simplex (\ref{symb}) can be transformed to that over a unit hyper\-cube. An appropriate change of variables is
$$
u_1=t_0t_1\cdots t_{2n-1}\,, \quad u_2=t_0t_1\cdots t_{2n-2}\,, \quad\ldots, \quad u_{2n-1}=t_0t_1\,, \quad u_{2n}=t_0\,,
$$
where $t_k\in [0,1]$. The corresponding Jacobian is equal to
$t_0^{2n-1}t_1^{2n-2}\cdots t_{2n-2}$. For $A_1$ this yields the
following symbol of two-cocycle:
\begin{equation}\label{c2}
    \hat \tau_2(p_1,p_2)=\omega(p_1, p_2)
\end{equation}
    $$
    \times\int_0^1dt_1\int_0^1 dt_0t_0  e^{i[\omega(p_0, p_1)(1-2t_0t_1)+\omega(p_0, p_2)(1-2t_0)+\omega(p_1, p_2)(1-2t_0+2 t_0t_1)]}\,.
$$

In the next section we will derive this integral representation for the FFS cocycle following a systematic procedure.

\section{Vasiliev Resolution}
Let $\Omega^\bullet$ denote the space of exterior differential forms whose  homogeneous elements are given by
\begin{equation}\label{a}
a=a(y,z)_{i_1\cdots i_q}dz^{i_1}\wedge \cdots \wedge dz^{i_q}\in \Omega^q\,.
\end{equation}
Here $a_{i_1\cdots i_q}(y,z)\in \mathbb{C}[y^1,\ldots,y^{2n},z^1,\ldots,z^{2n}]$. The $\mathbb{C}$-linear space $\Omega^\bullet$ can be endowed with the structure of an associative graded algebra with respect to the following $\ast$-product:
\begin{equation}\label{star}
a\ast b =a \exp i \left[\frac{\stackrel{\leftarrow}{\partial}}{\partial y^i}\pi^{ij}\left( \frac{\stackrel{\rightarrow}{\partial}}{\partial y^j}+\frac{\stackrel{\rightarrow}{\partial}}{\partial z^j}\right)\right]b\,,\qquad \forall a,b\in \Omega^\bullet\,.
\end{equation}
In other words, the product combines the exterior product of forms with the Moyal $\ast$-product of $y$'s and $z$'s and the grading is given by the form degree.

The algebra $\Omega^\bullet$ enjoys the involutive automorphism
$$
\tilde{a}(y,z;dz)=a(-y,-z;-dz)\,,\qquad \tilde{\tilde {a}}=a\,,\qquad \widetilde{a\ast b}=\tilde{a}\ast\tilde{b}\,.
$$
Using this automorphism, we can make  the algebra $\Omega^\bullet$ into a bimodule over itself with the following left and right actions:
$$
b\mapsto a\ast b\ast \tilde{c}\qquad \forall a,b,c \in \Omega^\bullet\,.
$$

Denote by $\hat\Omega^\bullet$ the completion of the space $\Omega^\bullet$. The elements of $\hat\Omega^\bullet$ are differential forms (\ref{a}) with coefficients being formal power series in $y$'s and $z$'s. The $\Omega^\bullet$-bimodule structure above extends naturally from the space $\Omega^\bullet$ to its completion $\hat\Omega^\bullet$.
Notice that $\Omega^\bullet$ contains the Weyl algebra $A_n=\mathbb{C}[y^1,\ldots,y^{2n}]$ as the subalgebra of $0$-forms that are independent of $z$'s.
Hence, we can think of $\hat{\Omega}^\bullet$ as a bimodule over $A_n$ as well.
This allows us to define the Hochschild complex  of the Weyl algebra $A_n$ with coefficients in $\hat\Omega^\bullet$. Each
$p$-cochain  $f\in C^p(A_n,\Omega^\bullet)$ is given by a $\mathbb{C}$-linear map $f: A_n^{\otimes p}\rightarrow \hat\Omega^\bullet$ and the action of the Hochschild differential $\partial: C^p\rightarrow C^{p+1}$  is defined by the usual formula
\begin{equation}\label{d-H}
(\partial f)(a_1,\ldots, a_{p+1})=a_1\ast f(a_2,\ldots,a_{p+1})
\end{equation}
$$
+\sum_{k=1}^{p}(-1)^{k+1} f(a_1,\ldots,a_k\ast a_{k+1},\ldots, a_{p+1})+(-1)^{p+1}f(a_1,\ldots,a_{p})\ast\tilde{a}_{p+1}\,.
$$
Clearly, the Hochschild complex $C^\bullet(A_n,\hat\Omega^\bullet)$ contains $C^\bullet(A_n, A_n^\ast)$ as subcomplex.

Now we observe  that the $A_n$-bimodule  $\hat\Omega^\bullet$ is actually a cochain complex with respect to the exterior differential $d:\hat\Omega^q\rightarrow \hat\Omega^{q+1}$. This gives one more differential on the cochain space $C^\bullet(A_n,\hat\Omega^\bullet)$. By definition,
\begin{equation}\label{d}
d:C^p(A_n,\hat \Omega^q)\rightarrow C^p(A_n,\hat \Omega^{q+1})\,,\qquad df=dz^i\wedge \frac{\partial f}{\partial z^i}(-1)^p\,.
\end{equation}
The sign factor $(-1)^p$ ensures that $ d\partial+\partial d=0$. Hence, the differential $d$ makes the Hochschild complex into the bicomplex $C^{\bullet,\bullet}=C^\bullet(A_n,\hat \Omega^\bullet)$ with
$$
\partial : C^{p,q}\rightarrow C^{p+1,q}\,,\qquad d: C^{p,q}\rightarrow C^{p,q+1}\,,\qquad p,q\geq 0\,.
$$
Associated to the bicomplex $C^{\bullet,\bullet}$ is the total complex $C^\bullet$, where $C^m=\bigoplus_{p+q=m} C^{p,q}$ and the differential $D:C^m \rightarrow C^{m+1}$ is given by the sum $D=d+\partial$.

Given the bicomplex $(C^{\bullet,\bullet}, \partial, d)$, we can define the pair of spectral sequences $\{'E^{\bullet,\bullet}_r\}$ and $\{''E^{\bullet,\bullet}_r\}$ converging to the cohomology of the total complex $H^\bullet_D(C)$.  As usual,
$$
'E_2^{p,q}=H^p_\partial H^q_d(C)\,,\qquad ''E_2^{p,q}=H^q_dH^p_\partial (C)\,.
$$
By the Poincar\'e Lemma the cohomology of the differential $d$ is concentrated in degree $0$ and
$$
H_d^0(C)\simeq C^\bullet(A_{n},A_{n}^\ast)\,.
$$
Hence,
\begin{equation}\label{HHH}
'E^{\bullet,\bullet }_2\simeq {} 'E^{\bullet,0}_2=H_\partial^\bullet H_d^0(C)\simeq HH^\bullet (A_{n}, A_n^\ast)\simeq HH^{2n}(A_n,A^\ast_n)\simeq \mathbb{C}\,.
\end{equation}
Here we made use of our knowledge of the Hochschild cohomology groups (\ref{HH}).
Thus,  the first spectral sequence collapses after the first step yielding $'E_2^{\bullet,\bullet}\simeq {}'E_\infty^{\bullet,\bullet}\simeq \mathbb{C}$. In other words,
\begin{equation}\label{HD}
H_D^\bullet(C)\simeq H^{2n}_D(C)\simeq\mathbb{C}
\end{equation}
and the one-dimensional space $H^{2n}_D(C)$ is generated by the FFS cocycle $\tau_{2n}$.

Consider now the second spectral sequence. The zero group $H_\partial^0(C^{\bullet,\bullet})$ is certainly nontrivial. By definition, it is represented by  forms
$\alpha\in \hat\Omega^\bullet$ satisfying  the equation
$$
b\ast \alpha-\alpha\ast \tilde{b}=0\qquad \forall b\in A_n\,.
$$
It is enough to check the last condition only for the generators of $A_n$. We have
\begin{equation}\label{eq}
y^j\ast \alpha+\alpha\ast y^j=0\quad \Leftrightarrow\quad 2y^j \alpha+i\pi^{jk}\frac{\partial \alpha}{\partial z^k}=0\,.
\end{equation}
The general solution to these equations is given by
\begin{equation}\label{aaa}
\alpha=e^{2i\omega(z, y)}g_{i_1\cdots i_q}(y)dz^{i_1}\wedge\cdots\wedge dz^{i_q}\in \hat \Omega^q\,,
\end{equation}
where $g$'s are  arbitrary formal power series in $y$'s. Then it
easy to see that $H^\bullet_dH^0_\partial(C)\simeq
H^{2n}_dH^0_\partial(C)\simeq\mathbb{C}$, where the group
$''E_2^{0,2n}=H^{2n}_dH^0_\partial(C)$ is generated by the
cocycle\footnote{Hint: all the coefficients of the form $d\alpha$,
where $\alpha$ is given by (\ref{aaa}), vanish at $y=0$.}
$$
\zeta=e^{2i\omega(z, y)} dz^{1}\wedge\cdots\wedge dz^{{2n}} \in \hat{\Omega}^{2n}\,.
$$
 As is seen, $\zeta\in C^{0,2n}$ is a $2n$-cocycle of the total complex, $D\zeta=0$. Since the total cohomology is known to be one-dimensional (\ref{HD}), the cocycle $\zeta$ must be cohomologous to $a\tau_{2n}$ for some $a\in \mathbb{C}$.  In other words, there exists $\xi\in C^{2n-1}$ such that
\begin{equation}\label{DFZ}
D\xi=\zeta-a\tau_{2n}\,.
\end{equation}
Expanding $\xi$ in homogeneous components,
$$
\xi=\xi_0+\xi_1+\cdots+\xi_{2n-1}\,,\qquad \xi_k\in C^{2n-1-k,k}\,,
$$
we can rewrite (\ref{DFZ}) as the system of descent  equations
\begin{equation}\label{dF}
\begin{array}{ll}
 d\xi_{2n-1}=\zeta\,,&\\
 d\xi_{2n-2}=-\partial \xi_{2n-1}\,,&\\
 \;\cdots&   \\
 d\xi_0=-\partial \xi_1\,,&\\
 \partial \xi_0=-a\tau_{2n}\,.&
 \end{array}
\end{equation}
In order to solve these equations we introduce the contraction homotopy operator $s:\hat\Omega^q\rightarrow\hat\Omega^{q-1} $ by the relation
$$
s a= \int_0^1dt\,t^{q-1}a(y,tz)_{i_1i_2\cdots i_q}z^{i_1}dz^{i_2}\wedge\cdots\wedge dz^{i_q}
$$
for  $a$ given by Eq. (\ref{a}). We have $sd+ds=p$, where $p:\hat \Omega^\bullet\rightarrow \hat \Omega^0$ is the canonical projection onto the subspace of $0$-forms. Using the operator $s$, we can solve successively all but one equations of  (\ref{dF}) as
$$
\begin{array}{l}
\xi_{2n-1}=s\zeta+d\phi_{2n-2}\,,\\
\xi_{2n-2}=(-s\partial)^2s\zeta+\partial \phi_{2n-2}+d\phi_{2n-3}\,,\\
\;\cdots\\
\xi_1=(-s\partial)^{2n-2}s\zeta+\partial \phi_{1}+d\phi_{0}\,,\\
\xi_0=(-s\partial)^{2n-1}s\zeta+\partial \phi_0-\varphi\,,
\end{array}
$$
for some $\phi_k\in C^{2n-k-2,k}$ and $\varphi\in C^{2n-1}(A_n,A_n^\ast)$ describing the general solution. Substituting $\xi_0$ into the last equation in (\ref{dF}),  we get the equality
\begin{equation}\label{Phi'}
a\tau_{2n}= \tau_{2n}' +\partial \varphi\,, \qquad \tau_{2n}' \equiv\partial(s\partial)^{2n-1}s\zeta\,.
\end{equation}
It remains to note that the cocycle $\tau_{2n}'$ is nontrivial and $a\neq 0$. Indeed, evaluating $\tau_{2n}'$ on the cycle $c_{2n}=1\otimes y^{1}\wedge\cdots\wedge y^{{2n}}$, we find
$$
  B(1,\tau'_{2n}(y^1\wedge \cdots \wedge y^{2n}))=\tau'_{2n}(y^1\wedge\cdots\wedge y^{2n})=(2n!)^{-1}\,,
$$
and hence $a =1$. The last computation is considerably simplified if one writes the Hochschild differential  as the sum $\partial=\partial_1+\partial_2$, where the operators
$\partial_1$ and $\partial_2$ are given by the first and second lines in (\ref{d-H}), respectively. It is easy to see that $\partial_2s+s\partial_2=0$ and $s^2=0$. This allows us to replace the full differential $\partial$ in (\ref{Phi'})  with $\partial_1$. We get
\begin{equation}\label{F''}
    \tau_{2n}' = \partial(s\partial_1)^{2n-1}s\zeta=\partial_1(s\partial_1)^{2n-1}s\zeta=(\partial_1s)^{2n}\zeta\,.
\end{equation}
(By construction, $\tau'_{2n}$ is independent of $z$; hence, we can evaluate it at $z=0$ and this yields the mid equality.) Formally, the cocycle (\ref{F''}) looks  like a coboundary. One should keep in mind, however, that its ``potential'' $\xi_0$ is a $z$-dependent function and not an element of $C^{2n-1}(A_n,A_n^\ast)$.
Expression (\ref{F''}) reproduces the FFS cocycle in the form of iterated integrals over a unit hypercube, c.f. (\ref{c2}).

The derivation above is essentially relied on the following
injective resolution  of the Hoch\-schild complex\footnote{ That is,
an injective resolution of the complex of vector spaces
$C^\bullet(A_n,A_n^\ast)$. Since
$H^\bullet_dH^0_\partial(C)\neq 0$, the exact sequence (\ref{VR})
is by no means a Cartan--Eilenberg resolution \cite[Ch.
XVII]{CartEil} as one might suspect.}:
\begin{equation}\label{VR}
  \xymatrix{ 0\ar[r]& C^\bullet(A_n,A_n^\ast)\ar[r]^{\varepsilon}& C^\bullet(A_n,\hat\Omega^0)\ar[r]^{d}& C^\bullet(A_n,\hat\Omega^1)\ar[r]^-{d} &\cdots}\,,
\end{equation}
with $\varepsilon$ being the natural embedding. For historical reasons explained in the Introduction we call the exact sequence (\ref{VR})  the {\it Vasiliev resolution} of the Hochschild  complex $C^\bullet(A_n,A_n^\ast)$.

\section{Twisted Bimodules and Smash Products}

The canonical generators of the Weyl algebra $A_n$ --
the formal variables $y^i$ -- span a $2n$-dimensional symplectic space
$(V,\omega)$ over $\mathbb{C}$. The symplectic group $Sp(n,\mathbb{C})$ acts on
$V$ by linear canonical transformations preserving $\omega$. The action of each
element $g\in Sp(2n,\mathbb{C})$ naturally extends to an automorphism of
the Weyl algebra $A_n$. We will write $a^g$ for the result of this action on
an element $a\in A_n$.   Let us assume that $g$ is a diagonalizable element of
$Sp(n,\mathbb{C})$. Then,  it is easy to see that $\mathrm{Im}(1-g)|_V$ is a
symplectic subspace of $V$ and we set $2k_g=\mathrm{rank}(1-g)|_V$.

To each automorphism $g$ of the Weyl algebra one can associate the $g$-twisted
bimodule $A_ng$ over $A_n$. This is defined as a linear space $A_n$ endowed with the following action of the Weyl algebra:
$$
a\rightarrow b\ast a\ast c^g\,,\qquad \forall a,b,c\in A_n\,.
$$

The dimensions of the Hochschild cohomology spaces $HH^\bullet(A_n,A_ng)$ were computed by Alev, Farinati, Lambre and Solotar \cite{AFLS} (see also \cite{Pinczon}).

\begin{theorem}[AFLS]\label{AFLS1} Let $g$ be a diagonalizable element  of $Sp(n,\mathbb{C})$ and $2k_g= \mathrm{rank} (1-g)|_V$, then
$$
HH^\bullet(A_n,A_ng)\simeq HH^{2k_g}(A_n,A_ng)\simeq \mathbb{C}\,.
$$
\end{theorem}

Notice that the isomorphisms (\ref{HH0}) and (\ref{HH}) follow from the AFLS theorem if we put $g=\pm 1|_V$.  The theorem also implies that the odd  cohomology
groups are all zero.

Our aim is to write an explicit formula for a nontrivial $2k_g$-cocycle, which we denote
by  $\tau_{g}$.

To do this, we only need to slightly adapt the Vasiliev resolution (\ref{VR}) to the $g$-twisted bimodules. As for the FFS cocycle,  we introduce the algebra of differential forms $\Omega^\bullet$ with the general element (\ref{a}) and the $\ast$-product (\ref{star}). The latter is obviously invariant under the automorphisms
$$
a^g(y,z;dz)=a(y^g,z^g;dz^g)\,,\qquad g\in Sp(2n,\mathbb{C})\,.
$$
(Here we just identify the linear spaces of $y$'s and $z$'s.) Twisting then the right action of $\Omega^\bullet$ on itself by $g$, we define the $g$-twisted
bimodule $\Omega^\bullet_g$ and its  completion $\hat\Omega^\bullet_g$. The natural embedding $A_n\subset \Omega^\bullet$  allows us to introduce the Hochschild complex $C^\bullet(A_n,\hat\Omega^\bullet_g)$, which is actually a bicomplex with respect to the Hochschild and the exterior differential (\ref{d}).
Writing this bicomplex as the complex of complexes
\begin{equation}\label{TVR}
  \xymatrix{ 0\ar[r]& C^\bullet(A_n,A_ng)\ar[r]^{\varepsilon}& C^\bullet(A_n,\hat\Omega_g^0)\ar[r]^{d}& C^\bullet(A_n,\hat\Omega_g^1)\ar[r]^-{d} &\cdots}
\end{equation}
augmented by $C^\bullet(A_n,A_ng)$, we get a natural generalization of the Vasiliev resolution (\ref{VR}) to the case of twisted bimodules.

Repeating the spectral sequence arguments of the previous Section,
one can see that the cocycle
$\tau_{g}$ is cohomologous (in the total complex) to a cocycle of
$C^{0,2k_g}$. The latter represents a nontrivial class in $H_d^{2k_g}H_\partial^0(C)$.

Recall that the group $H_\partial^0(C)$ is identified with the subspace of invariants of the bimodule $\Omega^\bullet_g$, i.e.,
$$
H^0_\partial(C)\ni a\quad \Leftrightarrow \quad b\ast a-a\ast b^g=0\qquad \forall b\in A_n\,.
$$
The last condition is enough to check only for the generators of $A_n$. This gives a system  of differential equations on $a$ similar to (\ref{eq}).
The general $q$-form satisfying these equations looks like:
\begin{equation}\label{center}
a=e^{i\omega_g}h(w)_{i_1\cdots i_q}dz^{i_1}\wedge\cdots\wedge dz^{i_q} \,,
\end{equation}
where
$$
w=z-(y-z)^g\,,\qquad \omega_g=\omega(z,y+(z-y)^g)\,,
$$
and $h$'s are arbitrary formal power series in $w$'s.

The next step is to calculate the $d$-cohomology in the space of differential forms  (\ref{center}).
We have
\begin{equation}\label{dw}
    d\omega_g=\omega(dz-dz^g, w)\,,
\end{equation}
so that the space $H_\partial^0(C)$ is  invariant under the action of $d$, as it must.
Let us introduce  the two- and $2k_g$-forms
$$
\omega (dz-dz^g,dz-dz^g)\,,\qquad
\varpi_g =\omega (dz-dz^g,dz-dz^g)^{k_g}\,.
$$
The rank of the two-form being $2k_g$, $\varpi_g\neq 0$.

\begin{proposition}\label{prop1}
The group $H_d^\bullet H_\partial^0(C)$ is generated by the $2k_g$-form
$$
\zeta_g=\varpi_g e^{i\omega_g}\,.
$$
\end{proposition}
\begin{proof}
Since $g$ is assumed to be diagonalizable, we can split the carrier symplectic space   as $V=V'\oplus V''$,
where $V'=\mathrm{Im}(1-g)|_V$ and $V''=\mathrm{Ker}(1-g)|_V$. Hence, each $v\in V$
can be uniquely decomposed as $v=v'+v''$ for $v'\in V'$ and $v''\in V''$. It is clear that $\omega(v',v'')=0$.
Let $y=(y', y'')$ and $z=(z',z'')$ be the sets of generators adapted to the splitting above.
The differential $d$ decomposes into the sum $d=d'+d''$, where $d'$ and $d''$ are exterior differentials
with respect to the variables $z'$ and $z''$.  It follows from (\ref{dw}) that $d''\omega_g=0$.
Applying the standard homotopy operator for the exterior differential $d''$ then shows that the
nontrivial cohomology is nested in the subspace of forms (\ref{center}) that depend neither  on $w''$
nor on $dz''$. In other words, the complex $(H_\partial^0(C), d)$ is homotopic to the complex $(K^\bullet, d')$ of forms
$$
a=e^{i\omega_g}h(w'; dz')\,.
$$
Multiplication by $e^{-i\omega_g}$ isomorphically maps the complex  $K^\bullet$ onto the complex $K_1^\bullet$ of forms $h(w';dz')$ with differential
$$
d_1=(dz'+dz'^g)^\alpha \wedge \frac{\partial}{\partial w'^\alpha}+\omega(d\tilde z,w')\wedge\,.
$$
Here $\tilde z=z'-z'^g$ and the index $\alpha=1,\ldots,2k_g$ labels the coordinates with respect to the $g$-diagonal bases in $V'$.
Since $\omega|_{V_1}$ is nondegenerate, there exists a matrix $\gamma=(\gamma^{\alpha\beta})$ such that
$$
d_1=e^{\gamma^{\alpha\beta}\frac{\partial^2}{\partial w'^\alpha\partial w'^\beta}}\omega(d\tilde z,w')\wedge e^{-\gamma^{\alpha\beta}\frac{\partial^2}{\partial w'^\alpha\partial w'^\beta}}\,.
$$
In other words, the differential $d_1$ is equivalent to $d_2=\omega(d\tilde z,w')\wedge$. The cohomology of the latter can easily be computed. Consider the operator
$$
\Delta=\pi^{\alpha\beta}i_{v_\alpha} L_{u_\beta}
$$
composed of the Lie and internal derivatives. Here $v_\alpha=\partial/\partial \tilde{z}^\alpha$, $u_\beta=\partial /\partial w'^\beta$, and $\pi$ is the bivector dual to the symplectic
form $\omega|_{V'}$.
One can check that  $\Delta d_2+d_2\Delta=N $, where the action of the differential operator $N$ on $q$-forms $h(w'; d\tilde z)$ is given by
$$
N=w'^\alpha \frac{\partial}{\partial w'^\alpha} +2k_g-q\,.
$$
The kernel of the operator $N$ is spanned by the form $\lambda=d\tilde z^1\wedge \cdots \wedge d\tilde z^{2k_g}$. The form $\lambda$ is clearly a nontrivial cocycle of $d_2$ (the coefficients of each $d_2$-exact form vanish at $w'=0$). Applying the inverse equivalence transformations yields then the form $e^{i\omega_g}\lambda$, which may differ from $\zeta_g=\varpi_g e^{i\omega_g}$ only by sign.
\end{proof}

By construction, the form $\zeta_g=\varpi_g e^{i\omega_g}$ is a cocycle of the total complex associated with
the bicomplex $C^\bullet(A_n,\hat\Omega_g^\bullet)$. Therefore, if nontrivial, it must be cohomologous to $\tau_{g}$ up to a normalization constant.
Writing down and solving the descent equations similar to (\ref{dF}), we set
\begin{equation}\label{t}
\tau_{g}=(\partial_1s)^{2k_g}\zeta_g\,.
\end{equation}
Recall that the operator $\partial_1$ is defined by the first line in \eqref{d-H}.

\begin{proposition} The Hochschild cocycle defined by Eq. (\ref{t}) is nontrivial.
\end{proposition}
\begin{proof} We will prove the statement by evaluating the cocycle $\tau_{g}$ on the dual cycle $c_{g}$. To write down the latter we introduce a $g$-diagonal symplectic basis in $V$ such that $y=(q^1,p_1,\ldots q^{n},p_n)$, $\omega(dy,dy)=\sum_{i=1}^n dp_i\wedge dq^i$, and $$y^g=(\lambda_1 q^1,\lambda_1^{-1}p_1,\ldots,\lambda_n q^n,\lambda_n^{-1}p_n)\,,$$
where the eigen values are ordered such that $\lambda_i\neq 1$ for $i=1,\ldots, k_g$ and $\lambda_i=1$ for $i>k_g$.
Given the bilinear form (\ref{B}), we can identify the bimodule dual to $A_ng$ with $\tilde{g}A_n$; here the left-twisting element is given by the product $\tilde{g}=gg_0$, where $g_0$ is the parity automorphism, $y^{g_0}=-y$. Then, it is not hard to check that
$$
c_{g}=\Theta\otimes p_1\wedge q^1\wedge\cdots \wedge p_{k_g}\wedge q^{k_g}\,,\qquad \Theta=e^{-i\sum_{i=1}^{k_g} \left(\frac{1+\lambda_i}{1-\lambda_i}\right)p_iq^i}\,,
$$
is a $2k_g$-cycle of the normalized Hochschild complex $\bar C_\bullet(A_n,\tilde gA_n)$.
The function $\Theta$ is chosen to satisfy the equations
$$
\Theta\ast q^i+\lambda_iq^i\ast\Theta=0\,,\qquad \Theta\ast p_i+\lambda^{-1}_ip_i\ast\Theta=0\,,\qquad i=1,\ldots, k_g\,.
$$
A direct computation shows that
$$
\tau_{g}(c_{g})=B(\Theta, \tau_{g}(p_1\wedge q^1\wedge\cdots\wedge p_{k_g}\wedge q^{k_g}))=(2k_g!)^{-1}\,.
$$
Hence, both $\tau_{g}$ and $c_{g}$ are nontrivial and generate the groups $HH^{2k_g}(A_n,A_ng)$ and $HH_{2k_g}(A_n,\tilde g A_n)$, respectively.
\end{proof}

\begin{remark}\label{R2}
The symplectic group $Sp(n,\mathbb{C})$ acts on $C^\bullet(A_n,A_ng)$ by the cochain transformations
$$
c^h(a_1,\ldots, a_k)=\big(c(a_1^{h^{-1}},\ldots, a_k^{h^{-1}})\big)^h\,,\qquad \forall h\in Sp(n,\mathbb{C})\,.
$$
If we treat the bases cocycles (\ref{t}) as a map taking each diagonalizable element $g\in Sp(n,\mathbb{C})$ to the cochain $\tau_g\in C^\bullet(A_n,A_ng)$, then one can easily check that this map
is equivariant with respect to the adjoint action of $Sp(n,\mathbb{C})$ on itself, namely,
$$
\tau_g^h=\tau_{hgh^{-1}}\,.
$$
Notice also that the element $hgh^{-1}$ is diagonalizable if $g$ was such and $k_g=k_{hgh^{-1}}$. In other words,
the elements of the conjugacy class $[g]$ are characterized by the same number $k_g$.

\end{remark}

Finally, we turn to the smash-product algebras. Let $G$ be a finite subgroup acting by automorphisms on a $\mathbb{C}$-algebra $A$ and let $\mathbb{C}[G]$ be the corresponding group algebra. Recall that the smash product
$A\rtimes G$ is the $\mathbb{C}$-algebra with the underlying space $A\otimes \mathbb{C}[G]$ and the product
$$
(a\otimes g)(b\otimes h)=a\ast b^g\otimes gh\,,\qquad \forall a,b\in A\,,
\quad \forall g,h\in G\,.
$$

The following statement relates the Hochschild cohomology of the smash product $A\rtimes G$ to the $G$-invariant cohomology of the algebra $A$.

\begin{lemma}\label{lemma}
Let $A$ be a $\mathbb{C}$-algebra together with an action of a finite group $G$. Then
$$
HH^\bullet (A\rtimes G, A\rtimes G)\simeq \big(\oplus_{g\in G}HH^\bullet (A, Ag)\big)^G\,.
$$
Here $Ag$ is the bimodule isomorphic to $A$ as a space where the left
action of $A$ is the usual one and the right action is the usual action
twisted by $g$.
\end{lemma}
For the proof see e.g. \cite[Lemma 9.3]{Etingof}.

We are going to apply this lemma to the case $A=A_n$ and $G$ is a finite subgroup of the symplectic group $Sp(n,\mathbb{C})$.
The group $G$ being finite,  any $g\in G$  satisfies $g^{|G|}=1$; and hence, $g$ is diagonalizable. This allows us to use the above results on the cohomology groups
$HH^\bullet (A_n,A_ng)$. The elements of the direct sum $\oplus_{g\in G}HH^\bullet(A_n,A_ng)$ are  generated by the cocycles
\begin{equation}\label{tg}
\tau_\gamma=\sum_{g\in G}\gamma(g)\tau_g\,,
\end{equation}
where $\tau_g$ are basis cocycles for the nonzero groups $HH^{k_g}(A_n,A_ng)\simeq \mathbb{C}$ and $\gamma: G\rightarrow \mathbb{C}$ is a complex function on $G$.
The $G$-invariance condition implies that $\tau_\gamma^h=\tau_\gamma$ for all $h\in G$. Taking into account Remark \ref{R2}, this is equivalent to the condition $\gamma(hgh^{-1})=
\gamma(g)$, that is, $\gamma: G\rightarrow \mathbb{C}$ is a class function on $G$.
In such a way we arrive at the following theorem proved in \cite{AFLS}.
\begin{theorem}[AFLS] \label{AFLS2}The cohomology space $HH^p(A_n\rtimes G,A_n\rtimes G)$ is naturally isomorphic to the space of conjugation invariant functions on the set $S_p$ of elements $g\in G$ such that $$\mathrm{rank} (1-g)|_V=p\,.$$
\end{theorem}

When applied to the case under consideration, the isomorphism established by Lemma \ref{lemma} admits a fairly explicit description. Given a $G$-invariant, normalized cocycle (\ref{tg}), define the cocycle
\begin{equation}\label{vtg}
\vartheta_\gamma=\sum_{g\in G}\gamma(g)\vartheta_g  g
\end{equation}
of the Hochschild complex $C^\bullet(A_n\rtimes G, A_n\rtimes G)$ by setting
\begin{equation}\label{tt}
\vartheta_g(a_1g_1\,,\ldots,a_{2k_g}g_{2k_g})=\tau_g(a_1,a^{g_1}_2,a_3^{g_1g_2},\ldots, a^{g_1\cdots g_{2k_g-1}}_{2k_g})\otimes g_1g_2\cdots g_{2k_g}\,,
\end{equation}
for all $a_i\in A_n$ and $g_i\in G$. The reader can check that $\vartheta_\gamma$ is a cocycle indeed. Furthermore, $\vartheta_\gamma=0$ whenever  some of its arguments belongs to $\mathbb{C}[G]$.

One could arrive at the cocycles (\ref{vtg}) starting from the injective resolution
\begin{equation}\label{VR3}
\xymatrix{  C^\bullet({A}_n\rtimes G,{A}_n\rtimes G)\ar[r]^{\varepsilon}&C^\bullet({A}_n\rtimes G,\hat\Omega^\bullet\rtimes G)}\,,
\end{equation}
where the action of $G$ naturally extends from $A_n$ to the differential forms (\ref{a}) and the action of exterior differential in $\hat\Omega^\bullet\rtimes G$ is defined as $d(a\otimes g)=da\otimes g$ for all $a\in\hat\Omega^\bullet$ and $g\in G$.

Considering the pair of spectral sequences canonically associated to the bicomplex $C^\bullet(A_n\rtimes G,\hat\Omega^\bullet\rtimes G)$, one can infer that
\begin{equation}\label{iso}
HH^p({A}_n\rtimes G,{A}_n\rtimes G)\simeq H^p_dH^0_\partial({A}_n\rtimes G,\hat\Omega^\bullet\rtimes G)\,,
\end{equation}
where the groups on the right are generated by the $G$-invariant forms
\begin{equation}\label{zg}
\zeta_\gamma=\sum_{g\in G}\gamma(g)\zeta_gg\,.
\end{equation}
Here the forms $\zeta_g$ are defined by Proposition \ref{prop1} and $\gamma$ is a class function on $G$.  (Note that the forms $\zeta_g$ are equivariant in the sense that $\zeta_g^h=\zeta_{hgh^{-1}}$ for all $h\in G$.) Applying  the operator $(\partial_1 s)^{p}$ to  a $p$-form in (\ref{zg}) yields then the desired $p$-cocycle
$$
\vartheta_\gamma=(\partial_1s)^{p}\zeta_\gamma\,.
$$

\begin{example} \label{Ex}By way of illustration, we apply the above constructions to a particular case of physical interest. Namely, we compute the Hochschild cocycles of the higher-spin algebra \cite{Vasiliev:1986qx} underlying $4d$ higher-spin theories.  The algebra in question is given by the smash product $\mathfrak{A}=A_2\rtimes G$, where
$$G= \mathbb{Z}_2\times \mathbb{Z}_2\subset Sp(1,\mathbb{C})\times Sp(1,\mathbb{C})\subset Sp(2,\mathbb{C})\,.$$

In order to make contact with the notation adopted in the physical literature let us also describe $\mathfrak{A}$ by generators and relations. The Weyl algebra $A_2$ is generated by four complex variables $\{y^\alpha, \bar y^{\dot\alpha}\}$, $\alpha,\dot\alpha=1,2$, with defining relations
$$
    [y^\alpha, y^\beta]=2i\epsilon^{\alpha\beta}\,,\qquad [y^\alpha,\bar y^{\dot\alpha}]=0\,,\qquad  [\bar y^{\dot\alpha},\bar y^{\dot\beta}]=2i\epsilon^{\dot\alpha\dot\beta}\,.
$$
Here $\epsilon$'s are the two-dimensional Levi-Civita symbols.  The complex conjugation acts on $y$'s as $(y^\alpha)^\ast=\bar y^{\dot\alpha}$.  The group $G$ is generated by the pair of symplectic  reflections $\kappa,\bar\kappa$:
$$
    \kappa y^\alpha=-y^\alpha \kappa \,,\qquad \kappa \bar y^{\dot\alpha}=\bar{y}^{\dot\alpha} \kappa\,,
    \qquad \bar\kappa \bar y^{\dot\alpha}=-\bar{y}^{\dot\alpha} \bar\kappa\,,
     \qquad  \bar\kappa y^\alpha=y^\alpha \bar\kappa \,,
$$
$$
\kappa^2=1\,,\qquad \bar\kappa^2=1\,,\qquad \kappa\bar\kappa=\bar\kappa \kappa\,,\qquad \kappa^\ast=\bar\kappa\,.
$$
The generators $\kappa$ and $\bar\kappa$ are usually called the Klein operators. Thus, the general element of the algebra $\mathfrak{A}$ is given by
$$
    a=a_1(y,\bar y)+a_2(y,\bar y)\kappa+a_3(y,\bar y)\bar\kappa +a_4(y,\bar y)\kappa\bar\kappa\,.
$$
The action of $\bar\kappa\kappa$ defines a $\mathbb{Z}_2$-grading  making $\mathfrak{A}$ into a superalgebra.

In view of the isomorphism $Sp(1,\mathbb{C})\simeq SL(2,\mathbb{C})$, one can think of the variables $y^\alpha$ and $\bar y^{\dot\alpha}$ as the left- and right-handed Weyl spinors of the $SL(2,\mathbb{C})$ group. Furthermore, the associated Lie algebra $L(\mathfrak{A})$ contains a finite-dimensional subalgebra generated by all real quadratic combinations of $y$'s and $\bar y$'s. This subalgebra is isomorphic to the Lie algebra $sp(2,\mathbb{R})\simeq so(3,2)$, i.e., the Lie algebra of isometries of $4d$ anti-de Sitter space. To some extent this explains the relevance of the algebra $\mathfrak{A}$ to $4d$ higher-spin field theories.

Theorem \ref{AFLS2} gives us information about the dimensions of the cohomology spaces $HH^\bullet(\mathfrak{A}, \mathfrak{A})$. We have  $G=\{1,\kappa,\bar\kappa, \kappa\bar\kappa\}$ and
$$S_0=\{1\}\,, \quad S_1=\{\varnothing\}\,, \quad S_2=\{\kappa,\bar\kappa\}\,,\quad S_3=\{\varnothing\}\,,\quad S_4=\{\kappa\bar\kappa\}\,.$$
Since the group $G$ is abelian, we immediately obtain
$$
HH^0(\mathfrak{A},\mathfrak{A})\simeq \mathbb{C}\,,\qquad HH^2(\mathfrak{A},\mathfrak{A})\simeq \mathbb{C}^2\,,\qquad HH^4(\mathfrak{A},\mathfrak{A})\simeq \mathbb{C}\,,
$$
and the other groups vanish.

As it usually is in deformations theory, the second cohomology group $HH^2(\mathfrak{A},\mathfrak{A})$ admits a direct physical interpretation: The corresponding basis cocycles generate the pair of cubic vertices in the Vasiliev equations for massless higher-spin fields \cite{Vasiliev:1988sa}.
To find explicit expressions for these vertices, we can use the Vasiliev resolution (\ref{VR3}) resulting in the isomorphism (\ref{iso}). By Proposition \ref{prop1} the group  $H_d^\bullet H^0_\partial(\mathfrak{A},\hat\Omega^\bullet\rtimes G)$ is generated by the forms
$$
1,\quad
e^{2i z_\alpha y^\alpha}\kappa dz_\alpha\wedge dz^\alpha,\quad  e^{2i\bar z_{\dot\alpha}\bar y^{\dot\alpha}}\bar \kappa d\bar z_{\dot\alpha}\wedge d\bar z^{\dot\alpha},\quad
e^{2i(z_\alpha y^\alpha+\bar z_{\dot\alpha}\bar y^{\dot\alpha})}\kappa\bar\kappa dz_\alpha\wedge dz^\alpha\wedge d\bar z_{\dot \alpha}\wedge d\bar z^{\dot\alpha},
$$
where $z_\alpha=z^\beta\epsilon_{\beta\alpha}$ and $\bar z_{\dot\alpha}=\bar z^{\dot\beta}\epsilon_{\dot\beta\dot\alpha}$.
Applying now the operator $(\partial_1s)^p$ for $p=0,2,4$ gives the Hochschild cocycles of $\mathfrak{A}$. In view of the property (\ref{tt}) these cocycles are completely determined
by their values on the subalgebra $A_2\in \mathfrak{A}$. The Weyl algebra $A_2$, being isomorphic to $A_1\otimes A_1$, is spanned by the polynomials $c(y,\bar y)=a(y)b(\bar y)$. We have
$$
\begin{array}{l}
\vartheta_0=1\,,\\[3mm]
\vartheta_2(a_1b_1, a_2b_2)=\tau_2(a_1,a_2)(b_1\ast b_2) \kappa\,,\\[3mm]
\bar\vartheta_2(a_1 b_1, a_2b_2)=(a_1\ast a_2)\tau_2(b_1,b_2) \bar\kappa\,,\\[3mm]
\vartheta_4 (c_1,c_2, c_3,c_4)=\tau_4(c_1,c_2,c_3,c_4)\kappa\bar\kappa\,,
 \end{array}
$$
where $\tau_2$ and $\tau_4$ are the FFS cocycles defined by Eq. (\ref{FFS}). In this form the cocycles $\vartheta_2$ and $\bar\vartheta_2$ were first found  in \cite{Vasiliev:1988sa}.

\end{example}

\subsection*{Acknowledgments} We are grateful to Vasiliy Dolgushev and Boris Feigin for useful discussions.  We also acknowledge a kind hospitality at the program ``Higher Spin Theory and Duality" (Munich, May 2-27, 2016) organized by the Munich Institute for Astro- and Particle Physics (MIAPP).

\appendix
\section{Geometrical Interpretation of the FFS Cocycle}\label{A}
By making change of integration variables, one can bring the symbol of the FFS cocycle (\ref{symb}) into the form
$$
   \hat\tau_{2n}(p_0,p_1,\ldots,p_{2n})= \int_{\mathbb{R}^{2n}} d^{2n}v\,\exp{i\left[2\omega(v, p_0+\cdots+p_{2n}) +\sum_{0\leq i<j\leq 2n} \omega(p_i, p_j)\right]}
   $$

   $$
   \times \Delta(v+p_0,v+p_0+p_1,\ldots,v+p_0+\ldots+p_{2n})\,,
$$
where $\Delta(v_0,\ldots,v_{2n})$
 is the characteristic function of the  oriented $2n$-simplex $\Sigma\subset \mathbb{R}^{2n}$ spanned by the vectors $v_0,v_1,\ldots, v_{2n}\in \mathbb{R}^{2n}$.
 By definition, the function $\Delta$ assumes only three different values $0, \pm 1$. The value $0$ says that the origin $0\in \mathbb{R}^{2n}$ lies outside the simplex $\Sigma$.   Otherwise $\Delta$ takes on value $1$ for the right simplices  and $-1$ for the left ones.
It follows from the definition that
\begin{itemize}
 \item $\Delta$ is totally anti-symmetric under permutations of its arguments;
    \item $\Delta$ is $Sp(n,\mathbb{R})$-invariant, i.e., any linear symplectic transform $v_i\rightarrow A v_i$ leaves it intact;
\item the Hochschild cocycle condition is equivalent to the identity
\begin{equation}\label{tid}
        \sum_{k=0}^{2n+1} (-1)^k\Delta(v_0,\ldots,\hat{v}_k,\ldots,v_{2n+1})=0\,.
\end{equation}
\end{itemize}

Geometrically, the identity expresses a simple fact that any polytope with $2n+2$ vertices in $\mathbb{R}^{2n}$ can be
cut into $2n+2$ simplices. Then the origin  $0\in \mathbb{R}^{2n}$ belongs to an even number of such simplices with appropriate sign factors and orientations ensuring pairwise cancelation.

From the algebraic viewpoint Eq.(\ref{tid}) means that the function $\Delta$ is a $2n$-cocycle of the Alexander-Spanier complex of $\mathbb{R}^{2n}$ with coefficients in $\mathbb{Z}$, see e.g. \cite{massey1978homology}.  The action of the coboundary operator on $p$-cochains is given by
$$
(\partial \varphi)(v_0,\ldots,v_{p+1})=\sum_{k=0}^{p+1} (-1)^k\varphi (v_0,\ldots,\hat{v}_k,\ldots,v_{p+1})\,.
$$
Notice that $\mathrm{supp}\,{\Delta}=\{0\}$, so that the cocycle $\Delta$ is compactly supported. The cocycle $\Delta$ can be formally represented as a coboundary, namely,
$$
\Delta =\partial \Delta_w \,,\qquad \Delta_w(v_0,\ldots, v_{2n-1})=\Delta(w, v_0,\ldots,v_{2n-1})\,.
$$
The cochain $\Delta_w$, however, is neither $Sp(n,\mathbb{R})$-invariant nor compactly supported. It is easy to see that for all nonzero $w\in \mathbb{R}^{2n}$, $\mathrm{supp}\, \Delta_w\simeq \mathbb{R}_+$. Actually, the $2n$-cocycle $\Delta$ generates the Alexander-Spanier cohomology  $H_c^{\bullet}(\mathbb{R}^{2n},\mathbb{Z})\simeq \mathbb{Z}$.

In the  special case of $A_1$ the identity (\ref{tid}) was first revealed in \cite{Vasiliev:1989xz}.

\providecommand{\href}[2]{#2}\begingroup\raggedright\endgroup

\end{document}